\documentclass[a4paper,11pt]{article}
\usepackage{fullpage}
\title{Covering Segments with Unit Squares}
\author{Ankush Acharyya, Subhas C. Nandy, Supantha Pandit, and Sasanka Roy}
\date{Indian Statistical Institute, Kolkata, India.
  \texttt{\{ankush\_r,nandysc,sasanka\}@isical.ac.in, pantha.pandit@gmail.com}}

\usepackage{amsmath,amssymb,amsthm}
\usepackage{graphicx}
\usepackage{epstopdf}
\usepackage[usenames,dvipsnames]{xcolor}
\usepackage[normalem]{ulem}
\usepackage{subcaption}
\usepackage{mathtools}
\usepackage{wrapfig}
\usepackage{lipsum}
\usepackage{mathrsfs}
\usepackage[T1]{fontenc}
\usepackage{algorithm}
\usepackage{algorithmic}
\usepackage[utf8]{inputenc}
\newtheorem{lemma}{Lemma}

\newtheorem{theorem}{Theorem}

\newcommand{\colb}[1]{{\color{blue}{\textbf{\textit{#1}}}}}
\newcommand{\colbm}[1]{{\color{blue}{\boldmath{\textit{#1}}}}}

\newcommand{\colnew}[1]{{\color{orange!80!black}{\textbf{\textit{#1}}}}}

\newcommand{\probname}[1]{{\color{orange!80!black}{\textbf{\textsc{#1}}}}}

\usepackage{pbox}
\usepackage{array}
\newcolumntype{P}[1]{>{\centering\arraybackslash}p{#1}}
\newcolumntype{M}[1]{>{\centering\arraybackslash}m{#1}}

\newtheorem{definition}{Definition}
\newtheorem{observation}{Observation}
\newtheorem{remark1}{Remark}

\usepackage{amssymb}

\newcommand{\np}{$\mathsf{NP}$}
\newcommand{\apx}{$\mathsf{APX}$}
\newcommand{\ptas}{$\mathsf{PTAS}$}

\begin{document}

\maketitle

\vspace{-0.1in}
\begin{abstract}
We study several variations of line segment covering problem with axis-parallel unit squares in $I\!\!R^2$. A set $S$ of $n$ line segments is given. The objective is to find the minimum number of axis-parallel unit squares which cover at least one end-point of each segment. The variations depend on the orientation and length of the input segments. We prove some of these problems to be \np-complete,  and give constant factor approximation algorithms for those problems. For some variations, we have polynomial time exact algorithms. For the general version of the problem, where the segments are of arbitrary length and orientation, and the squares are given as input, we propose a factor 16 approximation result based on multilevel linear programming relaxation technique, which may be useful for solving some other problems. Further, we show that our problems have connections with the problems studied by Arkin et al. \cite{Arkin2015} on conflict-free covering problem. Our \np-completeness results hold for more simplified types of objects than those of Arkin et al. \cite{Arkin2015}.  

\vspace{.5cm}
\noindent {\bf Keywords:} Segment cover, unit square, \np-hardness,  linear programming, approximation algorithms, \ptas.
 \end{abstract}

 \vspace{-0.1in}
\section{Introduction}  \vspace{-0.1in} 
\label{intro}
Covering is a well-studied problem in computer science and has applications in diverse settings. Here a universe, and a collection of subsets of the universe are given as input. A minimum number of subsets need to be picked to cover the universe. The general version of this problem is $\mathsf{NP}$-complete \cite{Karp1972}. Many researchers studied different variants of this problem. In this paper, we study different interesting variations of line segment covering problem. 

The motivation of studying this problem comes from its applications to network security \cite{Kobylkin2016}. Here, the objective is to check the connectivity or security of a physical network. A set of physical devices is deployed over a geographical area. These devices are communicated to each other through physical links. The objective is to check the security of the network by placing minimum number of devices which can sense all the links whose at least one end-point lies inside a desired geometrical object (circle/square) centered around it. This problem can be modelled as {\it line segment covering problem}, where the links can be interpreted as segments and the objects can be interpreted as unit squares. Note that, the links are considered as straight lines. In \cite{Kobylkin2016}, several other applications are also stated.

Let $S=\{s_1,s_2,\ldots,s_n\}$ be a set of line segments in $I\!\!R^2$. We say that an axis-parallel square $t$ \colb{covers} a line segment $s \in S$ if $t$ contains at least one \colb{end-point} of $s$. In this paper, we deal with two classes of covering problem: i) \colb{continuous}, and ii) \colb{discrete}.

 \probname{Continuous Covering Segments by Unit Squares ({\textit{CCSUS}}):} Given a set $S$ of segments in $I\!\!R^2$, the goal is to find a set $T$ of unit squares which covers all segments in $S$, and the cardinality of the set $T$ is minimum among all possible sets of unit squares covering $S$.

 \probname{Discrete Covering Segments by Unit Squares ({\textit{DCSUS}}):} Along with the set $S$ of segments, here we are given a set $T$ of unit squares in $I\!\!R^2$, and the goal is to find a subset $T'\subseteq T$ of minimum cardinality which can cover all the segments in $S$.

We study the following variations of covering problem for line segments which are classified depending upon their lengths and orientations.

\noindent \colb{\textsc{\textbf{Continuous Covering}}}
\begin{itemize}
 \item[\colnew{$\blacktriangleright$}] \colnew{{\textit{CCSUS\text{-}H1\text{-}US}}:} Horizontal unit segments inside a unit height strip.
 \item[\colnew{$\blacktriangleright$}] \colnew{{\textit{CCSUS\text{-}H1}}:} Horizontal unit segments. 
 \item[\colnew{$\blacktriangleright$}] \colnew{{\textit{CCSUS\text{-}HV1}}:} Horizontal and vertical unit segments.
 \item[\colnew{$\blacktriangleright$}] \colnew{{\textit{CCSUS\text{-}ARB}}:} Segments with arbitrary length and orientation.
\end{itemize}
 \colb{\textsc{\textbf{Discrete Covering}}}
\begin{itemize}
 \item[\colnew{$\blacktriangleright$}] \colnew{{\textit{DCSUS\text{-}ARB}}:} Segments with arbitrary length and orientation.
\end{itemize}

We define some terminologies and definitions  used in this paper. We use \colb{segment} to denote a line segment, and \colb{unit square} to denote an axis-parallel unit square. For a given non-vertical segment $s$, we define \colbm{$l(s)$} and \colbm{$r(s)$} to be the left and right end-points of $s$. For a vertical segment $s$,  $l(s)$ and $r(s)$ are defined to be the end-points of $s$ with highest and lowest $y$-coordinates respectively. The \colb{center} of a square $t$ is the point of intersection of its two diagonals. We use \colbm{$t(a,b)$} to denote a square of side length $b$ and whose center is at the point $a$. Further, we define \colbm{$right$-$half(t(a,b))$} to be the portion of $t(a,b)$ to the right of the vertical line passing through the point $a$. 

\begin{definition} \label{def}
Two segments in $S$ are said to be \colb{independent} if no unit square can cover both the segments.  A subset $S'\subseteq S$ is said to be an \colb{independent set} if every pair of segments in $S'$ is independent. A subset $S'\subseteq S$ of segments is said to be \colb{maximal independent set} if for any $s\in S\setminus S'$, $S' \cup \{s\}$ is not an independent set.
\end{definition}

\vspace{-0.1in}
\subsection{Connection with the paper of Arkin et al. \cite{Arkin2015}} \vspace{-0.1in}

We point out an interesting connection between this paper and the paper of Arkin et al. \cite{Arkin2015}. They studied a family of covering problem, 
called the \colb{conflict-free covering}. Given a set $P$ of $n$ color classes, where each color class contains exactly two points, the goal is 
to find a set of conflict-free objects of minimum cardinality which covers at least one point from each color class. An object is said to be \colb{conflict-free} if it contains at most one point from each color class. They looked at both discrete (where the covering objects are given as a part of the input) and continuous (where the covering objects can be placed anywhere in the plane) versions of conflict-free covering problem. When the points are on real line and the covering objects are intervals, both discrete and continuous versions of conflict-free covering problem are studied by them. Further, when same colored points are either horizontally or vertically separated by unit distance and the covering objects are unit squares, they studied the continuous version of conflict-free covering problem.

Instead of line segments, if we consider a given set of colored line segments where each segment is of different color and restrict each unit square to be conflict-free, then our definition of covering all the objects is equivalent to the conflict-free covering problem of Arkin et al. \cite{Arkin2015}. 

\vspace{-0.1in}

\subsection{Known results} \vspace{-0.1in}

Arkin et al. \cite{Arkin2015-1,Arkin2015} showed that, both discrete and continuous versions of conflict-free covering problem are \np-complete where the points are on a real line and objects are intervals of arbitrary length. These results are also valid for covering arbitrary length segments on a line with unit intervals (see Appendix \ref{ABC}).  They provided factor $2$ and factor $4$ approximation algorithms for the continuous and discrete versions of conflict-free covering problem with arbitrary length intervals respectively. If points of same color class are either vertically or horizontally unit separated, then they proved that the continuous version of conflict-free covering problem with axis-parallel unit square is \np-complete and proposed a factor $6$ approximation algorithm. Finally, they remarked the existence of a polynomial time dynamic programming based algorithm for the continuous version of conflict-free covering problem with unit intervals where the points are on a real line and each pair of same color points is unit separated.
Recently, Kobylkin \cite{Kobylkin2016} studied the problem of covering the edges of a given straight line embedding of a planar graph by minimum number of unit disks, where an edge is said to be covered by a disk if any point on that edge lies inside that disk. They proved \np-completeness results for some special graphs. A similar study is made in \cite{Reddy2016}, where a set of line segments is given, the objective is to cover these segments with minimum number of unit disks, where the covering of a segment by a disk is defined as in \cite{Kobylkin2016}. They studied both the discrete and continuous versions of the problem. For continuous version, they proposed a \ptas~where the segments are non-intersecting. For discrete version, they showed that the problem is \apx-hard.

\vspace{-0.1in}

\subsection{Our contributions} \vspace{-0.1in} 

In Section \ref{continuous}, we first propose an $O(n \log n)$ time greedy algorithm for \textit{CCSUS-H1-US} problem. This is used to propose a factor 2 approximation result for the \textit{CCSUS-H1} problem.  

Arkin et al. \cite{Arkin2015} showed that continious version of conflict-free covering problem of points by unit squares is \np-complete, where each pair of same color points is either horizontally or vertically unit distance apart. They also proposed a factor $6$ approximation algorithm for this problem. We show that the 
\textit{CCSUS-H1} problem is \np-complete. Thus our \np-completeness reduction works for more simplified types of objects than those of Arkin et al. \cite{Arkin2015}. In addition, we propose  an $O(n \log n)$ time factor $3$ approximation algorithm for the \textit{CCSUS-HV1} problem.  Finally, we provide a \ptas~for \textit{CCSUS-HV1} problem. We also give an $O(n \log n)$ time factor 6 approximation algorithm for the \textit{CCSUS-ARB} problem.

In Section \ref{discrete}, we give a polynomial time factor 16 approximation algorithm for the \textit{DCSUS-ARB} problem. It uses multiple levels of LP-relaxation, and finally an LP-based approximation algorithm. This method is of independent interest since it may be used to get approximation algorithm for some other problem.  The running time of our algorithm fully depends on the running time of solving linear programs. Getting an algorithm  with approximation factor better than 16 for the \textit{DCSUS-ARB} problem remains an interesting open question.

\vspace{-0.2in}
\section{Continuous covering}\label{continuous} \vspace{-0.1in}
Here, the segments are given, and the objective is to place minimum number of unit squares for covering at least one end-point of all the segments.  

\vspace{-0.1in}

\subsection{\textit{CCSUS-H1-US} problem} \label{CCSUS-H1-US}  \vspace{-0.1in}

Below, we give an $O(n\log n)$ time greedy algorithm for the \textit{CCSUS-H1-US} problem. Let $S$ be a set of $n$ horizontal unit segments inside a horizontal unit strip. Start with an empty set $T'$. Sort the segments in $S$ from left to right with respect to their right end-points. Repeat the following steps until $S$ becomes empty. Select the first segment $s\in S$ which is not yet covered by the last added square in $T'$. Place a unit square $t$ inside the strip aligning its left boundary at $r(s)$, and mark all the segments that are covered by $t$. Put $t$ in $T'$. Finally, return  the set $T'$ as the output. Using standard analysis of covering points on real line by unit intervals (See Cormen et al. \cite{Cormen2009}, Exercise $16.2$-$5$), we can prove the following theorem.

\begin{theorem} \label{th1}
The worst case time complexity of our algorithm for the \textit{CCSUS-H1-US} problem 
is $O(n\log n)$.
\end{theorem}

\vspace{-0.1in}

\subsection{\textit{CCSUS-H1}  problem}
\vspace{-0.1in}
Here,  we prove that \textit{CCSUS-H1} is \np-complete. Next, we propose an $O(n\log n)$ time factor 2 approximation algorithm for this problem.

\vspace{-0.1in}
\subsubsection{{\boldmath{\np}}-completeness}
\vspace{-0.1in}
The hardness of this problem is proved by a reduction from the rectilinear version of planar 3 SAT \colb{(RPSAT(3))} problem \cite{Knuth1992}, which is known to be \np-complete.

\probname{RPSAT(3) \cite{Knuth1992}:} Given a 3 SAT problem $\phi$ with $n$ variables and $m$ clauses, where the variables are positioned on a horizontal line and each clause containing 3 literals is formed with three vertical line segments and one horizontal line segments. Each clause is connected with its three variables either from above or from below such that two no two line segments corresponding to two different variables intersect. The objective is to find a satisfying assignment of $\phi$. See Figure \ref{defn-planar-3sat} for an instance of {\it RPSAT(3)} problem. Here the solid (resp. dotted) {\it vertical segment} attached to the horizontal line of a clause represents that the corresponding variable appears as a positive (resp. negative) literal in that clause.

\begin{figure}[htbp]
\begin{subfigure}[t]{1in}
\centering
\includegraphics[scale=0.35]{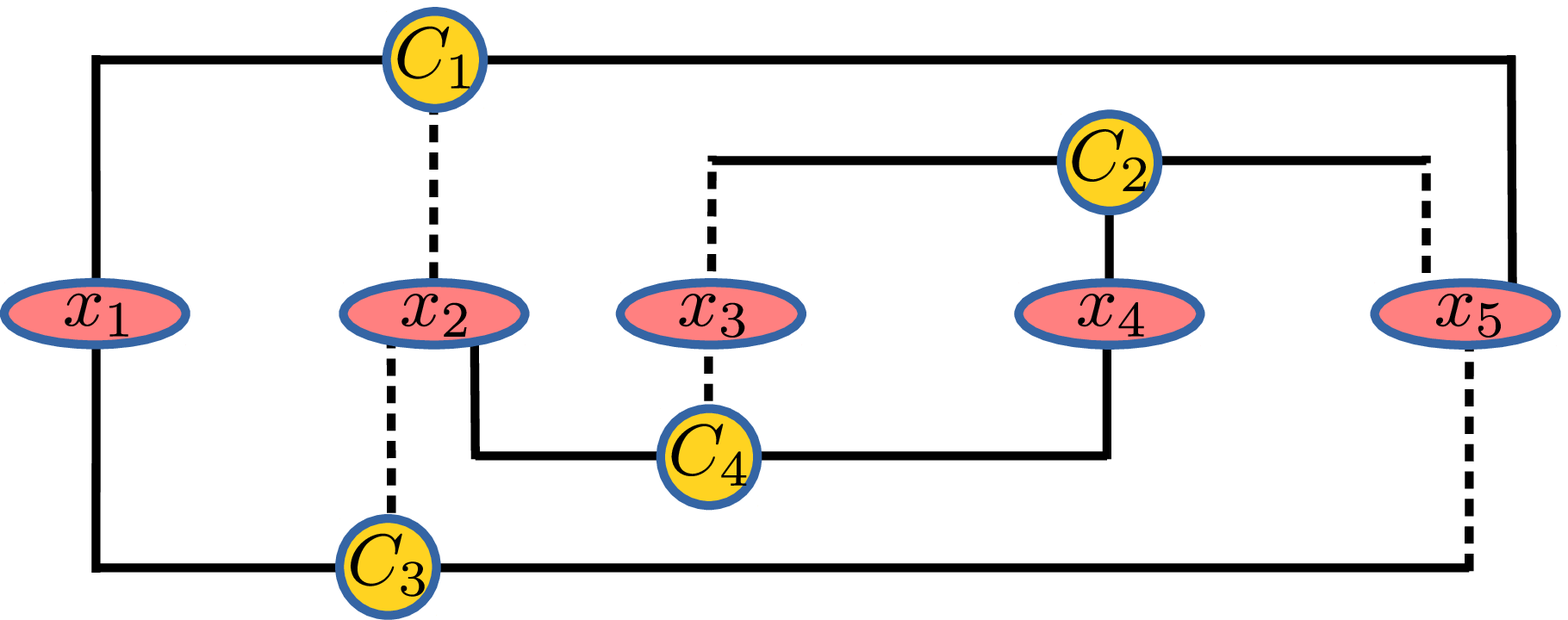}
\caption{ }
\label{defn-planar-3sat}
\end{subfigure}
\hspace{7.5cm}
\begin{subfigure}[t]{1in}
\centering
\includegraphics[scale=.35]{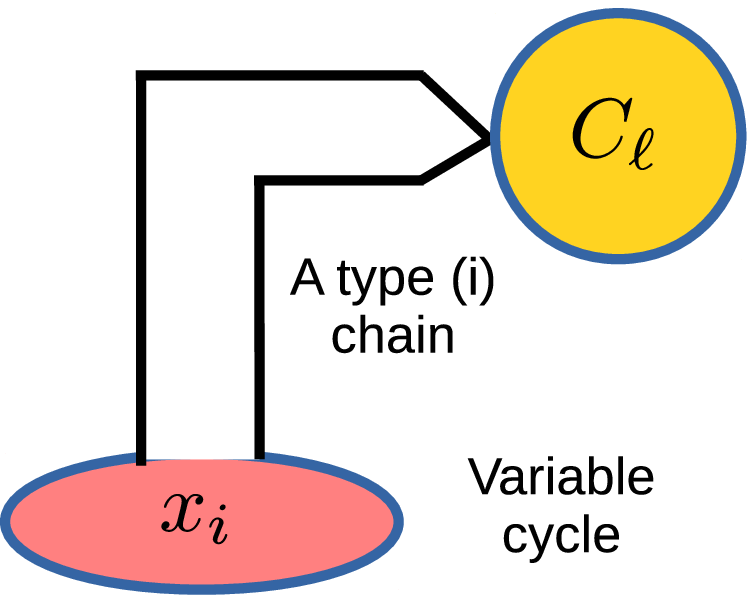}
\caption{ }
\label{fig-a-chain}
\end{subfigure}
\caption{(a) {\it RPSAT(3)} representation.  (b) Connection of a cycle and a chain.}
\label{fig-defn-planar-3-sat}
\vspace{-0.1in}
\end{figure}

We first describe the construction of an instance $I$ of \textit{CCSUS-H1} problem from an instance $\phi$ of {\it RPSAT(3)} problem. Next we validate the construction. 

Let $\{x_1, x_2, \ldots, x_n\}$ be $n$ variables and $\{C_1,C_2,\ldots,C_m\}$ be $m$ clauses of $\phi$. Here we describe the construction for the clauses connecting to the variables from above. A similar construction can be done for the clauses connecting to the variables from below. 

Let $d$ be the maximum number of vertical segments connected to a single variable from different clauses either from above or from below. Assume, $\delta = 4d + 3$.
Each variable gadget for $x_i$ may consist of a single \colb{cycle} and at most $2d$ number of \colb{chains}.  The {\it cycle} consists of $2\delta$ unit horizontal segments $\{s_1^i,s_2^i,\ldots,s_{2\delta}^i\}$ in two sides of a horizontal line (see Figure \ref{fig-var-gadget}). The segments $\{s_1^i,s_2^i,\ldots,s_{\delta}^i\}$ are above the horizontal line and the segments $\{s_{\delta + 1}^i,s_{\delta +2}^i,\ldots,s_{2\delta}^i\}$ are below the horizontal line. 
The chains correspond to the vertical segments connecting a variable $x_i$ with the clause containing it. There are three types of \colb{chains}: (i) `` \includegraphics[scale=.06]{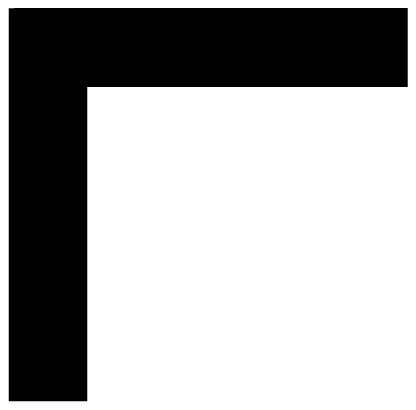} '', (ii) `` \includegraphics[scale=.06]{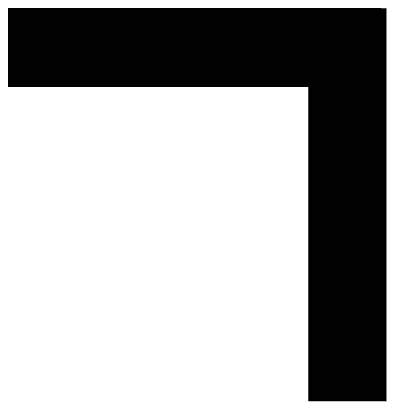} '', and (iii) `` \includegraphics[scale=.06]{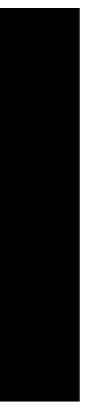} '' (see Figure \ref{defn-planar-3sat}). The gadget corresponding to three types of chains are shown in Figures \ref{chain-1}, \ref{chain-2}, and \ref{chain-3} respectively. The chains are connected to the cycle, and together it forms a chain of \colb{big-cycle} (see Figure \ref{fig-a-chain}).   It needs to mention that the number of segments is not fixed for every chain, even  for similar chains of different clauses. Note that, at the joining point (to construct a big-cycle) we slightly \colb{perturb}  two unit segments little upward.

\begin{figure}[t]
\begin{center}
\includegraphics[scale=.25]{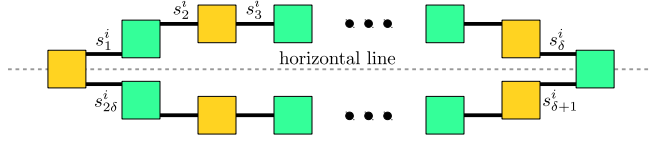}
\end{center}\vspace{-0.1in}
\caption{ Gadget for a variable $x_i$.}
\label{fig-var-gadget}
\vspace{-0.1in}
\end{figure}

\begin{figure}[h]
\begin{subfigure}[t]{1in}
\centering
\includegraphics[scale=.23]{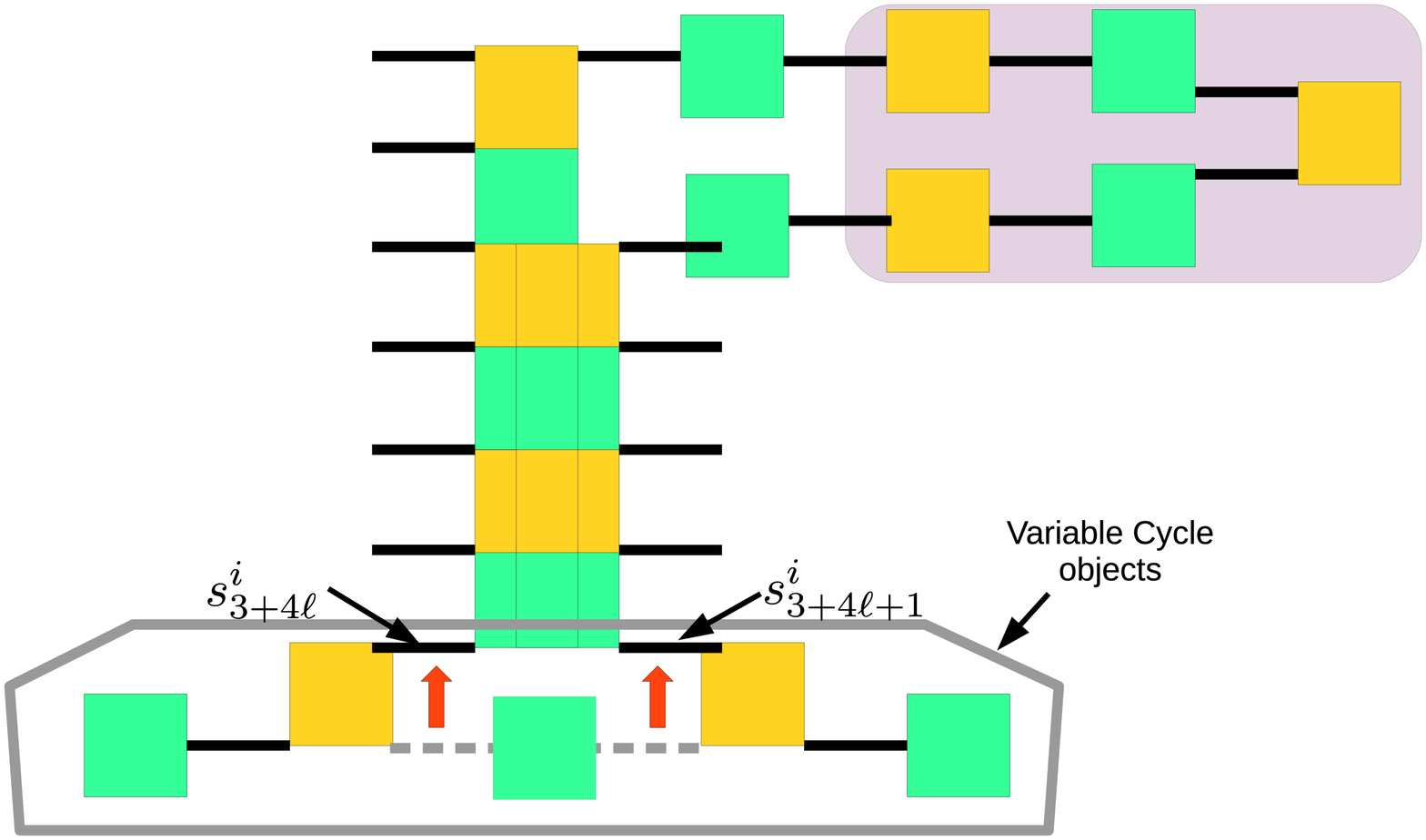}
\caption{ }
\label{chain-1}
\end{subfigure}
\hspace{5cm}
\begin{subfigure}[t]{1in}
\centering
\includegraphics[scale=.23]{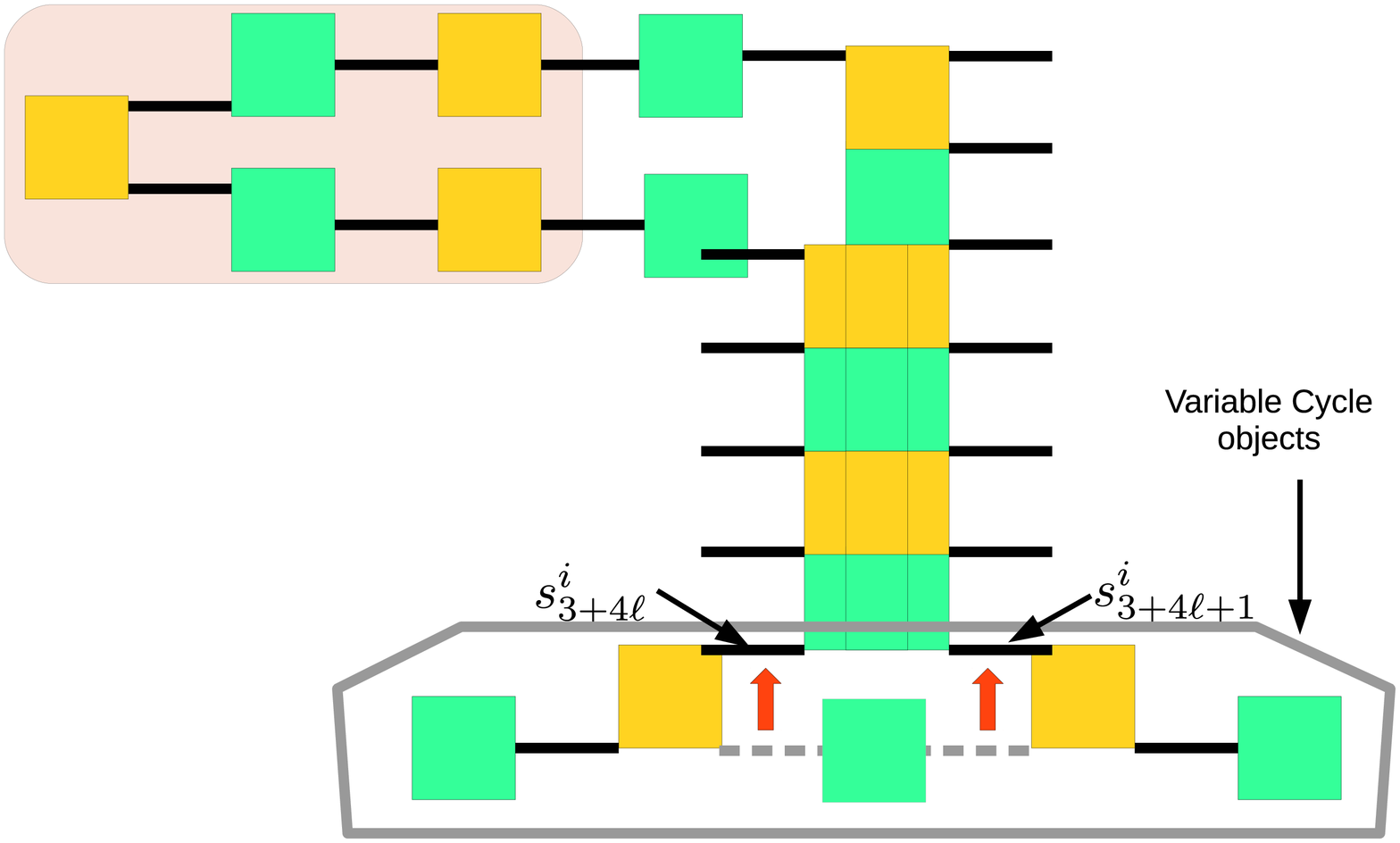}
\caption{ }
\label{chain-2}
\end{subfigure}

\begin{subfigure}[t]{1in}
\includegraphics[scale=.23]{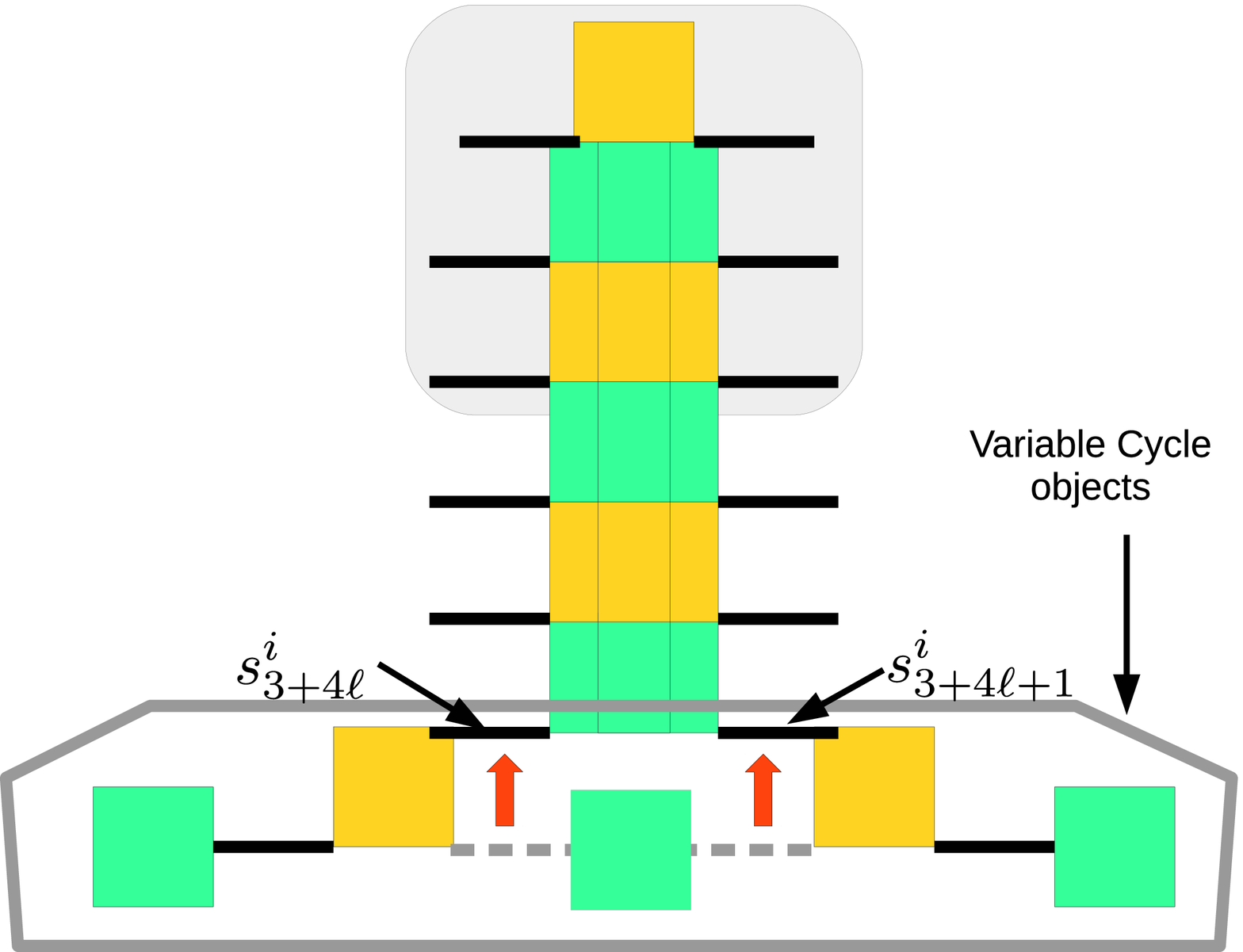}
\caption{ }
\label{chain-3}
\end{subfigure}
\hspace{5cm}
\begin{subfigure}[t]{1in}
\includegraphics[scale=.23]{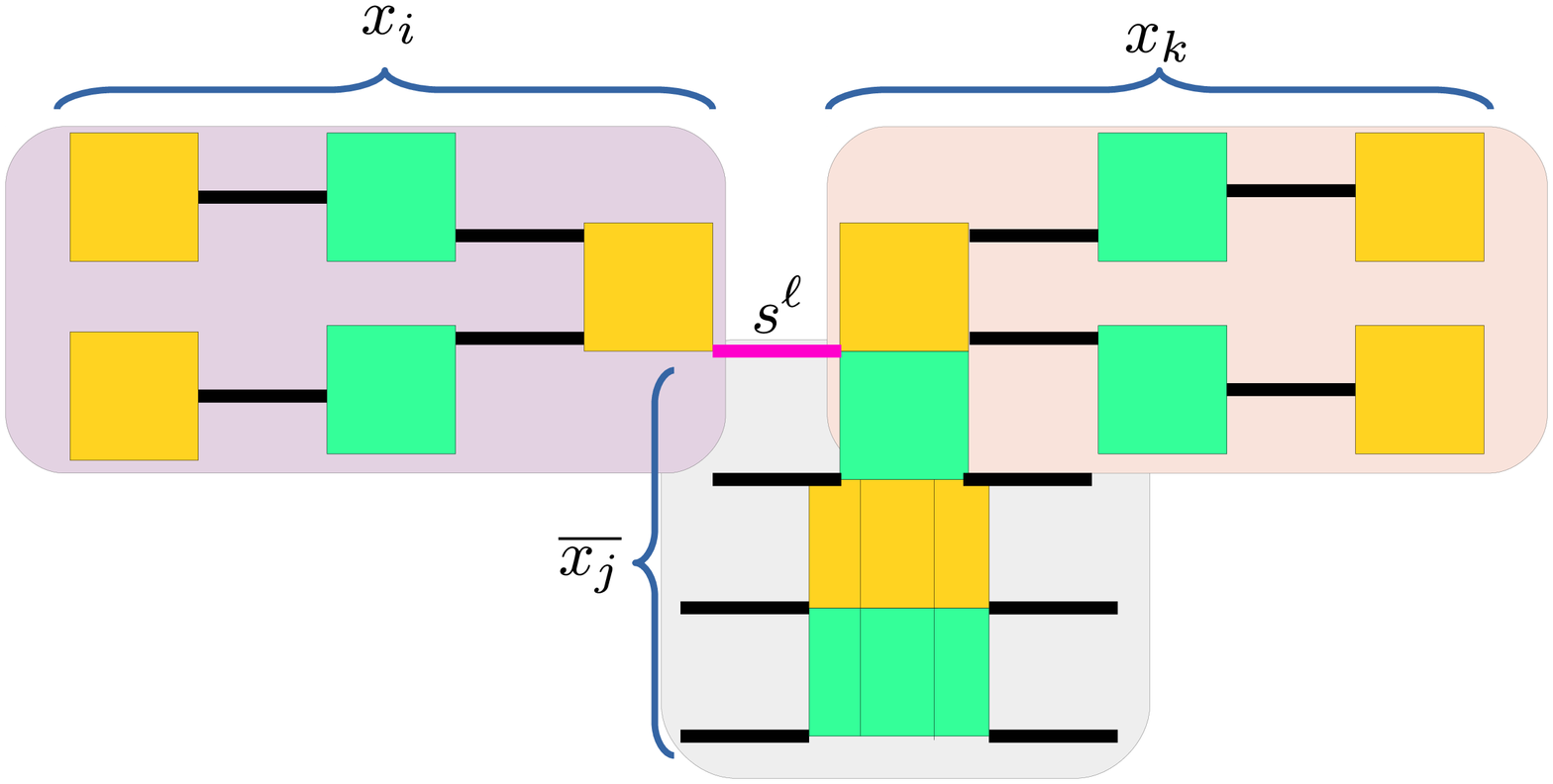}
\caption{ }
\label{chain-clause-junction}
\end{subfigure}
\caption{(a) Gadget for type (i) chain.  (b) Gadget for type (ii) chain. (c) Gadget for type (iii) chain. (d) Demonstration of clause-segment $s^\ell$ corresponding to the clause $C_\ell = (x_i\vee \overline{x_j} \vee x_k)$; here the shaded portions from the parts (a), (b) and (c) are shown to demonstrate the connection of $s^\ell$ with the variables in $C_\ell$.}
\label{chains}
\vspace{-0.1in}
\end{figure}

Let $0,1,2,\ldots,\kappa$ ($\kappa \leq d$) be the left to right order of the vertical segments corresponding to the clauses which are connected to the gadget 
corresponding to the variable $x_i$. Consider the $\ell$-th clause $C_\ell$ in this order. If $x_i$ is a positive literal, then the segments $s_{3+4\ell}^i$ and $s_{3+4\ell+1}^i$ are perturbed (moved upward as shown using upward arrow in Figures \ref{chain-1}, \ref{chain-2}, \ref{chain-3}) to connect the corresponding chain of $C_\ell$ with the cycle of variable $x_i$. Otherwise, If $x_i$ is a negative literal, then the segments $s_{3+4\ell+1}^i$ and $s_{3+4\ell+2}^i$ are perturbed.

Note that, the squares are not given as a part of the input. In the Figures \ref{fig-var-gadget}, \ref{chain-1}, \ref{chain-2}, and \ref{chain-3} a possible set of unit squares are also depicted. Each square can cover exactly two segments. Therefore, we have the following observation:
\begin{observation}\label{obb}
Exactly half of the squares (either all {\it green} or all {\it yellow}) can cover all the segments in the big-cycle corresponding to the variable $x_i$. This solution represents the truth value ({\it yellow} for true and {\it green} for false) of the corresponding variable $x_i$. 
\end{observation}

Further, for the clause $C_\ell$, we take a single unit horizontal segment $s^\ell$ that connects the chain corresponding to three variables. This is referred to as a \colb {clause-segment}. The placement of $s^\ell$ is shown in Figure \ref{chain-clause-junction}. Note that, in order to maintain the alternating green and yellow vertical layers in a variable gadget we may need to reduce the distance between two consecutive vertical layers of squares. But, the segments are placed sufficiently apart so that no unit square can cover more than two segmnts from a variable gadget. As the number of segments ($Q$) considering all variable gadgets, is even, we need exactly $\frac{Q}{2}$ squares to cover them. Now, if a clause $C_\ell$ is {\it satisfiable} then at least one square connected to $s^\ell$ will be chosen, and hence $s^\ell$ will be covered; if $C_\ell$ is {\it not satisfiable} then the square adjacent to $s^\ell$ of each variable chain will not be chosen in the solution, and hence we need one more square to cover $s^\ell$ (see Figure \ref{chain-clause-junction}). Thus, we have the following result, which leads to Theorem \ref{thh}.
\begin{lemma}
The given {\it RPSAT(3)} formula is satisfiable if  the number $N$ of squares needed to cover all the unit segments in the construction is exactly $N_0=\frac{1}{2}(\sum_{i=1}^n Q_i)$,  where $Q_i$ is the number of squares in the big-cycle corresponding to the gadget of the variable $x_i$. If the formula is not satisfiable then $N > N_0$, 
\end{lemma}

\begin{theorem}\label{thh}
\textit{CCSUS-H1} is \np-complete.
\end{theorem}

\subsubsection{Appriximation algorithm} \label{CCSUS-H1}

Let $S$ be a set of unit horizontal segments on the plane. We first partition the whole plane into a set of $\ell$ disjoint  unit height horizontal strips $H_1, H_2, \ldots, H_\ell$. Let $S_i \in S$ be the set of segments in the strip $H_i$, for $i=1, \ldots, \ell$. Clearly, $S_i\cap S_j =\emptyset$, for $i \neq j$. Now we have the following observation. 

\begin{observation} Any unit square cannot cover two segments, one from $S_i$ and the other from $S_j$ where $j-i\geq 2$, where $j>i, \text{ for } i= 1, \ldots,\ell-2, \text{ and }j= 3, \ldots, \ell$. 
\end{observation}

We calculate the minimum  number of unit squares covering $S_i$ (the segments in each $H_i$) using the algorithm in Section \ref{CCSUS-H1-US}. Let $Q_i$ be the set of squares returned for $H_i$ in our algorithm. Let $Q^{odd}= \{Q_1\cup Q_3\cup \ldots\}$ and $Q^{even}= \{Q_2\cup Q_4 \cup \ldots\}$ be the optimum solutions for 
the segments odd and even numbered strips respectively. We have $Q^{odd} \cap Q^{even}=\emptyset$,  and we report $Q=Q_{odd}\cup Q_{even}$. Let $OPT$ be a minimum sized  set of unit squares covering $S$. Now, $|OPT| \geq \max(|Q^{odd}|,|Q^{even}|)$. Thus, $|Q|=|Q_{odd}|+|Q_{even}| \leq 2|OPT|$. Since $S_i \cap S_j =\emptyset$ and the time for computing $Q_i$ is $O(|S_i|\log |S_i|)$ (by Theorem \ref{th1}), the overall running time of the algorithm is $O(n\log n)$. Thus, we have the following theorem.

\begin{theorem} \label{th2} A $2$-factor approximation result for the \textit{CCSUS-H1} problem can be computed in $O(n\log n)$ time. 
\end{theorem}

\subsection{\textit{CCSUS-HV1}  problem} \label{factor4algo} \vspace{-0.1in}
Here, we have both horizontal and vertical segments in $S$ which are of unit length. An easy way to get a factor 4 approximation algorithm for this problem is as follows. Let $S=S_H\cup S_V$, where $S_H$ and $S_V$ are the sets of horizontal and 
vertical unit segments respectively. We already have a factor 2 approximation algorithm for covering the members in $S_H$ (see Theorem \ref{th2}). The same algorithm works for $S_V$. Let $Q_H$ and $Q_V$ be the set of squares returned by our algorithm for covering $S_H$ and $S_V$ respectively. If $OPT_H$ and $OPT_V$ are the optimum solution for $S_H$ and $S_V$ respectively, and $OPT$ be the overall optimum solution for $S_H\cup S_V$, then $|OPT| \geq |OPT_H|$ and $|OPT| \geq |OPT_V|$.
Further, $|Q_H| \leq 2 |OPT_H|$ and $|Q_V| \leq 2 |OPT_V|$. Thus, $|Q_H|+|Q_V| \leq 2 |OPT_H|+2 |OPT_V|\leq 4|OPT|$.

We now propose a factor 3 approximation algorithm for this problem using sweep-line technique. During the execution of the algorithm, we maintain a set of segments 
$LB$ such  that no two of the members in $LB$ can be covered by an unit square. For each segment in $S$ we maintain a flag variable; its value is 1 or 0 depending 
on whether it is covered or not by the chosen set of squares corresponding to the members in $LB$. We also maintain a range tree $\cal T$ 
with the end-points of the members in $S$. Each element in $\cal T$ has a pointer to the corresponding element in $S$. In Algorithm \ref{algo-fact-3}, 
we describe the algorithm.

\begin{algorithm}[ht]
\begin{center}
\begin{algorithmic}[1]
\small
\STATE {\bf Input:} A set $S$ of $n$ horizontal and vertical unit segments.
\STATE {\bf Output:} A set $OUTPUT$ of unit squares which covers $S$.
\STATE $OUTPUT = \emptyset$; $LB=\emptyset$
\STATE sort the unit segments in $S$ from top to bottom according to their $r(.)$-values
\STATE (* see Section \ref{intro} for the definition of $r(s)$-values for the segments $s\in S$ *)  
\FOR {each segment $s \in S$ in order}
  \IF {$flag(s)=0$}
    \STATE insert  $s$ in $LB$; set $flag(s)=1$
     \IF {$s$ is horizontal}
         \STATE $m=3$, and define three unit squares $\{t_1,t_2,t_3\}$ as shown in Figure \ref{chsus-hv1-1}
         \STATE insert $t_1,t_2,t_3$ in $OUTPUT$ 
      \ELSIF {$s$ is vertical}
         \STATE $m=2$, and define two unit squares $\{t_1,t_2\}$ as shown in Figure \ref{chsus-hv1-2}
         \STATE insert $t_1,t_2$ in $OUTPUT$
     \ENDIF
     \FOR {$i=1,\ldots, m$}
           \STATE perform range searching with $t_i$
           \STATE for each element $\alpha$ of $\cal T$ in $t_i$ observe the corresponding element $s' \in S$
           \STATE set $flag(s')=1$; delete both the end-points of $s'$ from $\cal T$ 
     \ENDFOR
  \ENDIF
\ENDFOR
\STATE {\bf Return} $OUTPUT$
\end{algorithmic}
\end{center}
\caption{Factor $3$ Algorithm for \textit{CHSUS-HV1}.}
\label{algo-fact-3}\vspace{-0.1in}
\end{algorithm}
\normalsize

\begin{figure}[t]
\begin{center}
\begin{subfigure}[t]{1in}
\includegraphics[scale=.35]{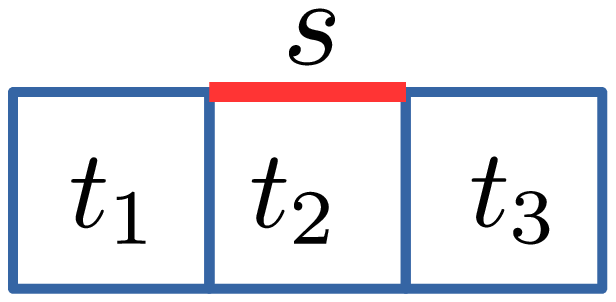}
\caption{ }
\label{chsus-hv1-1}
\end{subfigure}
\hspace{3.5cm}
\begin{subfigure}[t]{1in}
\includegraphics[scale=.35]{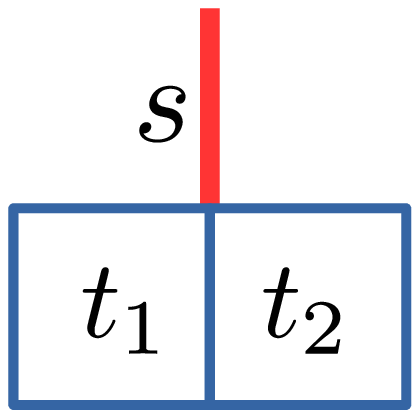}
\caption{ }
\label{chsus-hv1-2}
\end{subfigure}
\end{center}\vspace{-0.1in}
\caption{(a) Placement of $3$ unit squares $t_1, t_2$, $t_3$ for a horizontal unit segment $s$. (b) Placement of $2$ unit squares $t_1$ and $t_2$ for a vertical unit segment $s$.}
\vspace{-0.1in}
\end{figure}

\begin{theorem} Algorithm \ref{algo-fact-3} produces a $3$-factor approximation result for the \textit{CHSUS-HV1} problem, and it runs in $O(n\log n)$ time using $O(n\log n)$ space.
\end{theorem}

\begin{proof} Let $OPT$ be an optimal set of unit squares covering the members in $S$. In each iteration, we add a segment $s$ to $LB$ only if none of its end-points 
is covered by any unit square in $OUTPUT$. Clearly,  $LB$ is a maximal independent set of segments of $S$ (see Definition \ref{def}) and hence $|LB|\leq |OPT|$. Further, for each segment added to $LB$, at most $3$ unit squares are added to $T'$. Hence, $|T'|  \leq 3|LB| \leq 3|OPT|$. Also, when the algorithm terminates the squares in $OUTPUT$ covers the segments $S$. By using range searching data structure Algorithm \ref{algo-fact-3} can run in $O(n\log n)$ time.
\end{proof}

\subsubsection{Polynomial time approximation scheme}

In this section, we propose a \ptas~for the \textit{CCSUS-HV1} problem using the \colb{shifting strategy} of Hochbaum and Maass \cite{Hochbaum1985}. We are given a set 
$S$ of $n$ horizontal and vertical unit segments. Enclose the segments inside a integer length square box $B$; partition $B$ into vertical strips of width $1$, and also  partition $B$ into horizontal strips of height $1$. 
We choose a constant $k$, and define a \colb{k-strip} which consists of at most $k$ consecutive strips. Now, we define the concept of  \colb{shifts}. We have $k$ different shifts in the vertical direction. Each vertical shift consists of some disjoint $k$-strips. In the $i$-th shift ($i=0,1,\ldots, k-1$), the first $k$-strip consists of $i$ unit vertical strips at extreme left, and then onwards each $k$-strip is formed with $k$ consecutive unit vertical strips. Similarly, 
$k$ shifts are defined in horizontal direction. Now consider {\boldmath{\colb{shift$(i,j)$}}} as the $i$-th vertical shift and $j$-th horizontal shift. This splits the box $B$ into rectangular \colb{cells} of size at most $k \times k$. The following observation is important to analyze the complexity of our algorithm.  

\begin{observation} \label{obj} An optimal solution contains squares such that one boundary of each of those squares is attached to an end-point of some segment or two
boundaries are attached to the end-point of two different segments.
\end{observation}

{\bf Justification:} Suppose the boundary of a square $t$ in the optimum solution does not pass through any end-point (see Figure \ref{chsus-hv1-ptas-1}). We can move $t$ 
vertically/horizontally to touch to an end-point of some segment (see Figure \ref{chsus-hv1-ptas-2}). This square $t$ can further be moved in the  direction orthogonal to 
the previous movement to touch end-point of some other segment provided such an end-point exists within distance 1 from one of the boundaries of $t$ in that direction 
(see Figure \ref{chsus-hv1-ptas-3}). Figure \ref{chsus-hv1-ptas-4} shows that a unit square in an optimum solution may touch an end-point of only one segment also.

\begin{figure}[htbp]
\begin{center}
\begin{subfigure}[t]{1in}\includegraphics[scale=.3]{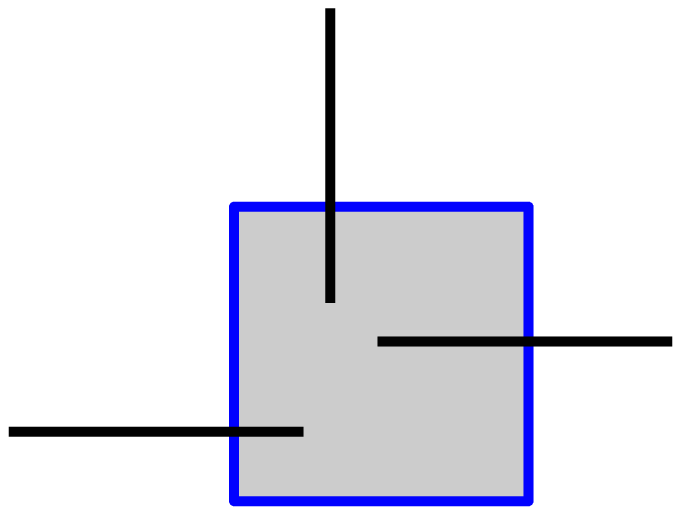}
\caption{ }
\label{chsus-hv1-ptas-1}
\end{subfigure}
\hspace{0.5cm}
\begin{subfigure}[t]{1in}
\includegraphics[scale=.3]{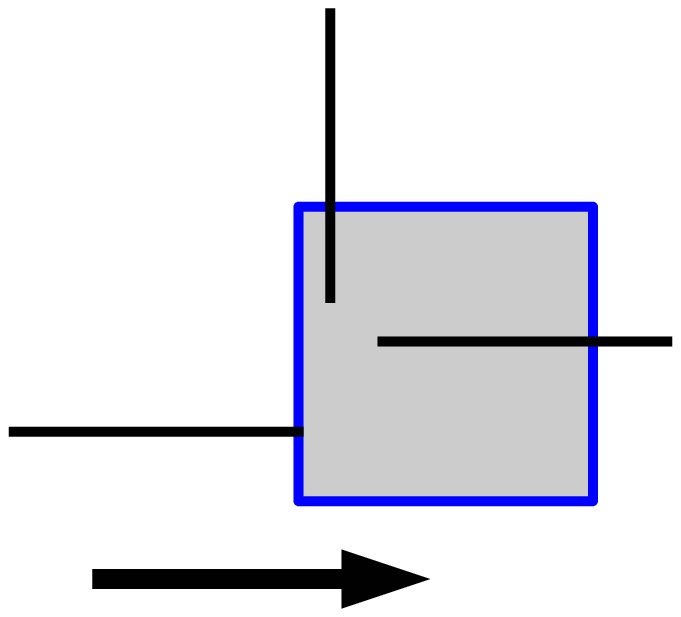}
\caption{ }
\label{chsus-hv1-ptas-2}
\end{subfigure}
\hspace{0.5cm}
\begin{subfigure}[t]{1in}\includegraphics[scale=.3]{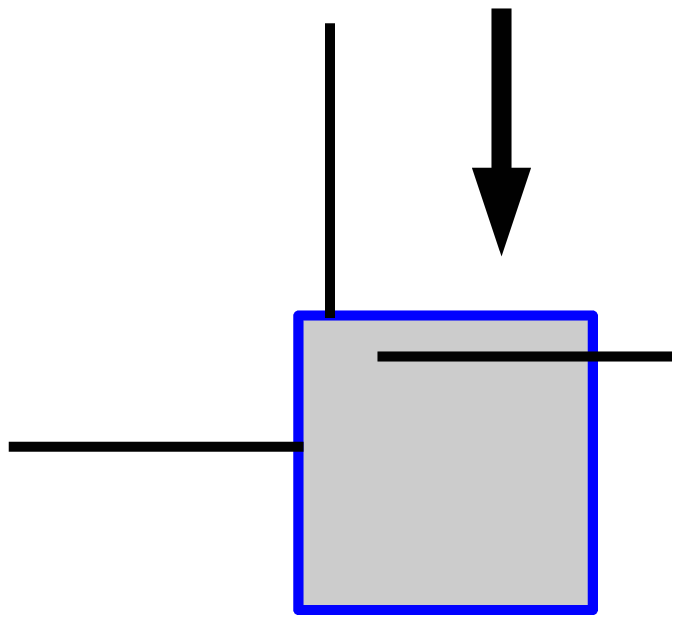}
\caption{ }
\label{chsus-hv1-ptas-3}
\end{subfigure}
\hspace{0.5cm}
\begin{subfigure}[t]{1in}\includegraphics[scale=.3]{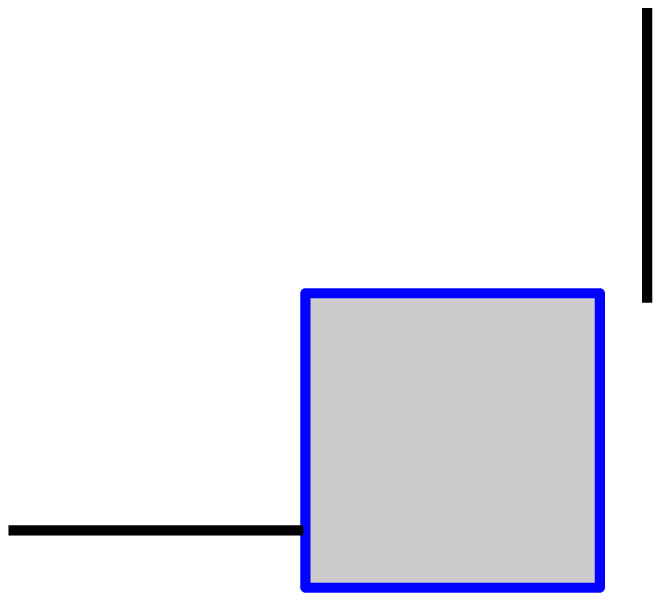}
\caption{ }
\label{chsus-hv1-ptas-4}
\end{subfigure}
\end{center}\vspace{-0.1in}
\caption{Justification of Observation \ref{obj}.}
\label{chsus-hv1-ptas-fig}
\vspace{-0.1in}
\end{figure}

\begin{lemma}\label{ptas-lemma} Finding a feasible solution for each $shift(i,j)$ require at most $O(n^{2k^2})$ time. 
\end{lemma}
\begin{proof} First consider a single cell $C$ of a particular $shift(i,j)$ which consists of $\chi$ many $1 \times 1$ cell, where $\chi$ is at most $k^2$. Let $n_C$ be 
the size of the set of segments which has a portion inside $C$. Observe that, at most $\chi$ unit squares can cover these $n_C$ segments in $C$. Again, at most $2n_C + 2\times {2n_C\choose 2}$ many positions are available for positioning 
the unit squares in an optimum solution of $C$ (see Observation \ref{obj} and Figure \ref{chsus-hv1-ptas-fig})\footnote{See Appendix \ref{appendix1} for the exact analysis of this count.}.    
We can use at most $\chi$ unit squares among $O({n_C^2)}$ possible positions to cover all the segments in $C$. 
Since optimal solution may be of any size in $\{1,2,\ldots, \chi\}$, we may need to consider   
$O(n_C^{2\chi})$ possible configurations to get the optimal solution. Since each segment can participate in at most two cells in $shift(i,j)$, 
we have $\sum_C n_C \leq 2n$. Thus, the time required for processing all the non-empty cells in $shift(i,j)$ requires at most $O(n^{2k^2})$.
\end{proof}

In our algorithm, for each $shift(i,j)$ we calculate optimal solution in each cell and combine them to get a feasible solution. Finally, return the minimum among these $k^2$ feasible solutions. 

Let $OPT$ be an optimum set of unit squares covering $S$, and $Q$ be a feasible solution returned by our algorithm described above. Now, we prove the following theorem.

\begin{theorem} $|Q| \leq (1+\frac{1}{k})^2|OPT|$ and the running time of the above algorithm is $O(k^2n^{2k^2})$.
\end{theorem}

\begin{proof} Let $Q_{ij}$ be the solution of our algorithm for $shift(i,j)$. Also, assume that $OPT_{ij}$ be the subset of squares in $OPT$ such that each of them  intersects the boundary of some cell in $shift(i,j)$. It can be shown that, $|\textsc{Q}_{ij}| \leq |OPT| + |OPT_{ij}|$ (Equation 2.3 of \cite{Hochbaum1985}). 
Now considering solutions for all $shift(i,j)$ for $1\leq i,j \leq k$, we have  

\centerline{$\displaystyle{\sum_{i=1}^k \sum_{j=1}^k |Q_{ij}|  \leq k^2 |OPT| + \sum_{i=1}^k \sum_{j=1}^k |OPT_{ij}|}.$}

Each horizontal (resp. vertical) line may be considered at most $k$ times during the $k$ vertical (resp. horizontal) shifts. 
Thus, each square intersecting a horizontal (resp. vertical) line may be counted at most $k$ times.  Thus we have, 

\centerline{$\displaystyle{\sum_{i=1}^k \sum_{j=1}^k |OPT_{ij}|  \leq k |OPT| + k |OPT|}.$ (see \cite{Goutam} for a similar analysis)}

Hence, $\displaystyle{\sum_{i=1}^k \sum_{j=1}^k |Q_{ij}|  \leq (k^2 +2k) |OPT|} \leq (k+1)^2|OPT|$, and 
finally, we have

$$\displaystyle{|Q| = \min_{i,j =1}^k\{|Q_{ij}|\} \leq \frac{\sum_{i=1}^k \sum_{j=1}^k |Q_{ij}|}{k^2}  \leq \left(1+\frac{1}{k}\right)^2 |OPT|}.$$

Using Lemma \ref{ptas-lemma}, we conclude that the running time of our algorithm is $O(k^2n^{2k^2})$.
\end{proof}

\subsection{\textit{CCSUS-ARB}  problem}

Mimicking the factor 2 approximation algorithm for the vertex cover problem of a graph, we can have a factor $8$ approximation algorithm for \textit{CCSUS-ARB} problem as follows. Next, we improve the approximation factor to $6$. 

\begin{observation}\label{greedy-obser} Let $s_1$ and $s_2$ be two segments in $S$. If none of the squares $t(l(s_1),2)$ and $t(r(s_1),2)$ covers $s_2$, then $s_1$ and $s_2$ are independent\footnote{Note that, $t(a,b)$ is a $b\times b$ square with center at $a$.}.
\end{observation}

As in Algorithm \ref{algo-fact-3}, here also we start with an empty set $OUTPUT$ and $LB$, and each segment in $S$ is attached with a flag bit. We maintain a range tree 
$\cal T$ with the end-points in $S$. Each time, an arbitrary segment $s \in S$ with $flag(s)=0$ is chosen, and inserted in $LB$. Its flag bit is set to 1. Insert 
four unit squares $\{t_1,t_2,t_3,t_4\}$ which fully cover the square $t(l(s),2)$ and four unit squares $\{t_1',t_2',t_3',t_4'\}$ which fully cover the square $t(r(s),2)$ in $OUTPUT$.
Remove all the segments in $S$ that are covered by $\{t_1, t_2,t_3,t_4,t_1',t_2',t_3',t_4'\}$ by performing range searching in $\cal T$ as stated in Algorithm 
\ref{algo-fact-3}.  The end-points of the deleted segments are also deleted from $\cal T$. This process is repeated until all the members in $S$ are flagged. Finally, return the set $OUTPUT$.

\begin{theorem} \label{th4} The above algorithm for \textit{CCSUS-ARB} problem runs in $O(n\log n)$ time, and produces a solution which is factor $8$ approximation of the optimal solution.
\end{theorem}
\begin{proof}
The approximation factor follows from the fact that $LB$ is a maximal independent set (see Observation \ref{greedy-obser}), and for each member in $LB$ we put 
8 squares in $OUTPUT$. The time complexity follows from that of Algorithm \ref{algo-fact-3}. 
\end{proof}

We now improve the approximation factor to $6$ using a sweep-line technique introduced in Biniaz et al. \cite{Biniaz2016}. We sort the segments in $S$ with respect to 
their left end-points\footnote{The left end-point of a vertical segment is its top end-point.}, and process the elements in $S$ in order. 
When an element $s\in S$  is processed, if $flas(s)=0$ then we put six squares, two of them covering the $1 \times 2$ rectangle 
$t_1 = right$-$half(t(l(s),2))$ and four of them covering the $2 \times 2$ square $t_2 =t(r(s),2)$. We also identify the segments in $S$ that are covered by 
$t_1$ and $t_2$, and their flag-bit is set to 1. Thus, we have the following result:
  
\begin{theorem} The above algorithm for the \textit{CCSUS-ARB} problem produces a factor $6$ approximation result in  
$O(n\log n)$ time.
\end{theorem}

\begin{remark1} The above algorithm gives a factor $3$ approximation result for the continuous covering horizontal segments of arbitrary length where the segments are inside a horizontal strip of unit height.
\end{remark1}

\section{Discrete covering: \textit{DCSUS-ARB}  problem}\label{discrete}

In this section, we give a 16 factor approximation algorithm for \textit{DCSUS-ARB} problem. Let $S$ be a set of $n$ arbitrary segments and $T$ be a set of $m$ unit 
squares. The algorithm runs in a series of \colb{steps}. In ($i+1$)-th step, we use linear programming to partition each subset of segments obtained in the $i$-th step 
into two disjoint subsets, and finally we obtain some subsets of $S$ such that for each subset the objective is to cover either left or right end-points of all the 
segments with the portions of unit squares which are above/below a horizontal line. 
To prove the approximation factor, we consider the following problem.

\probname{Covering points by unit width rectangles abutting $x$-axis (\textit{Restricted-Point-Cover}):} Given a set $P$ of points in $I\!\!R^2$ and a set $\cal R$ of unit width rectangles such that bottom boundary of each member in $\cal R$ coincides with the $x$-axis, find a subset of $\cal R$ of minimum cardinality to cover all the points in $P$.

\begin{lemma}\label{bansal-lemma} If $\mathtt{Z}_{RPC}$ be the standard {\it ILP} formulation of the \textit{Restricted-Point-Cover} problem, $OPT_{RPC}^I$ and $OPT_{RPC}^F$ are the optimum solutions of $\mathtt{Z}_{RPC}$ and its {\it LP}-relaxation respectively, then  
$OPT_{RPC}^I \leq 2 OPT_{RPC}^F$. 
\end{lemma}

In Section 3.2 of \cite{Bansal2010}, Bansal and Pruhs showed that $OPT_{RPC}^I \leq \alpha OPT_{RPC}^F$ for some positive constant $\alpha$ for a more generic version of this problem. We show that (in Appendix \ref{appendix2}) in our simplified case Lemma \ref{bansal-lemma} follows.

Let $\mathtt{Z}_\nu$ be an {\it ILP}. Denote \colbm{$\overline{\mathtt{Z}}_\nu$}, to be the {\it LP}-relaxation of 
$\mathtt{Z}_\nu$. Define \colbm{$OPT^I_\nu$} and \colbm{$OPT^F_\nu$} as the optimal solution of $\mathtt{Z}_\nu$ and $\overline{\mathtt{Z}}_\nu$ respectively.  
We first describe the different steps of the algorithm and finally establish the approximation factor. 

\noindent {\bf Step 1:} Let $T_1 \in T$ (resp. $T_2 \in T$) be the set of all squares which cover the left (resp. right) end-points of the segments in $S$. Now for each square $t_i \in T_1$, select a binary variable $x_i$, and for each square $t_j \in T_2$, select a binary variable $y_j$. Now create an $ILP$, $\mathtt{Z}_0$ as follows.

\vspace{-5mm}
\setcounter{equation}{0}
\begin{align*}
\mathtt{Z}_0: \text{ } \displaystyle{\min \sum_{i\mid t_i\in T_1} x_i + \sum_{j \mid t_j\in T_2} y_j} \hspace{3.3cm}
\end{align*} 
\vspace{-.5cm}
\begin{align*}
\text{s.t.} &\sum_{i \mid l(s_k)\in t_i\in T_1} x_i + \sum_{j\mid r(s_k)\in t_j \in T_2} y_j \geq 1  ~\forall~ k\mid s_k \in S; ~~x_i, y_j \in \{0,1\} 
~\forall~ i\mid t_i\in T_1 ~\&~ j\mid t_j\in T_2 \end{align*}
\vspace{-.3cm}

After solving this {\it ILP}, the value of $x_i$ = 1 or 0 depending on whether the square $t_i$ is in an optimal solution or not. Similarly, $y_j$ = 1 or 0 if the square $t_j$ is in an optimal solution or not.  

We solve the corresponding $LP$, $\overline{\mathtt{Z}}_0$. Now, create two partitions  $S_1 \subseteq S$ and $S_2 \subseteq S $ as follows. $S_1$ consists of those segments $s_k$ such that  $\sum_{i|l(s_k)\in t_i \in T_1} x_i \geq 1/2$, and $S_2$ consists of those segments $s_\ell$ for which $\sum_{j|l(s_\ell)\in t_j\in T_2} y_i \geq 1/2$. Now consider two $ILP$'s $\mathtt{Z}_1$ and $\mathtt{Z}_2$ as follows.

\vspace{-2mm}
\noindent \begin{minipage}{.5\textwidth}
\begin{align*}
\mathtt{Z}_1: \text{ } \displaystyle{\min \sum_{i\mid t_i\in T_1} x_i} \hspace{3.3cm}
\end{align*} 
\vspace{-.5cm}
\begin{align*}
 ~~~\text{s.t.} &\sum_{i \mid l(s_k)\in t_i\in T_1} x_i \geq 1,  ~~\forall~ k\mid s_k \in S_1 
\end{align*}
\vspace{-.5cm}
\begin{align*}
&x_i\in \{0,1\} ~~~ \forall i \mid t_i\in T_1
\end{align*}
\end{minipage}
\vline
\begin{minipage}{.5\textwidth}
\begin{align*}
\mathtt{Z}_2: \text{ } \displaystyle{\min \sum_{j \mid t_j\in T_2} y_j} \hspace{3.3cm}
\end{align*} 
\vspace{-.5cm}
\begin{align*}
 ~~~\text{s.t. } &\sum_{j\mid r(s_k)\in t_j\in T_2} y_j \geq 1,  ~~\forall~ k\mid s_k \in S_2 
\end{align*}
\vspace{-.5cm}
\begin{align*}
&y_j \in \{0,1\} ~~~~\forall ~ j \mid t_j\in T_2
\end{align*}
\end{minipage}

Then by an analysis identical to Gaur et al. \cite{Gaur_2002}, we conclude that $OPT_1^F + OPT_2^F \leq 2 OPT_0^F \leq 2 OPT_0^I.$

\noindent {\bf Step 2:} Observe that, both the $ILP$'s, $\mathtt{Z}_1$ and $\mathtt{Z}_2$ are the problems of covering points by unit squares. Consider the  covering problem $\mathtt{Z}_1$ \footnote{More precisely, $\mathtt{Z}_1$ (resp. $\mathtt{Z}_2$) corresponds to the covering problem of left (resp. right) end points of the segments in $S_1$ (resp. $S_2$) by unit squares.}. Divide the plane into unit strips by drawing horizontal lines. Observe that, no unit square in $T_1$ can intersect more than one line.  Partition $T_1$ into two sets $T_{11}$ and $T_{12}$, where $T_{11}$ consists of all squares which intersect even indexed lines and $T_{12}$ consists of all squares which intersect odd indexed lines. Define binary variables $x_i$ for $t_i \in T_{11}$ and $y_j$ for $t_j \in T_{12}$. Then, $\mathtt{Z}_1$ is equivalent to the following $ILP$.

\vspace{-5mm}
\begin{align*}
\mathtt{Z}_1: \text{ } \displaystyle{\min \sum_{i\mid t_{i}\in T_{11}} x_{i} + \sum_{j \mid t_{j}\in T_{12}} y_{j}} \hspace{3.3cm}
\end{align*} 
\vspace{-.5cm}
\begin{align*}
\text{s.t.} \hspace{-0.2in} &\sum_{i \mid l(s_k)\in t_{i}\in T_{11}} \hspace{-0.08in} x_{i} + \hspace{-0.2in} \sum_{j\mid l(s_k)\in t_{j}\in T_{12}} \hspace{-0.08in} y_{j} \geq 1  ~~\forall~ k\mid s_k \in S_1; ~~  
x_{i}, y_{j} \in \{0,1\} ~~\forall ~ i \mid t_{i}\in T_{11} ~\& ~ j \mid t_{j}\in T_{12}
\end{align*}

\vspace{-0.1in}
We solve $\overline{\mathtt{Z}}_1$. Now, create two groups $S_{11}$ and $S_{12}$. $S_{11}$ consists of those segments $s_k$ such that  $\sum_{i|l(s_k)\in t_{i} \in T_{11}} x_{i} \geq 1/2$, and $S_{12}$ consists of those segments $s_\ell$ for which $\sum_{j|l(s_\ell)\in t_{j}\in T_{12}} y_{i} \geq 1/2$.
Again consider two $ILP$'s, $\mathtt{Z}_{11}$ and $\mathtt{Z}_{12}$ as follows.

\begin{minipage}{.5\textwidth}
\begin{align*}
\mathtt{Z}_{11}: \text{ } \displaystyle{\min \sum_{i\mid t_{i}\in T_{11}} x_{i}} \hspace{3.3cm}
\end{align*} 
\vspace{-.5cm}
\begin{align*}
~~~\text{s.t. } &\sum_{i \mid l(s_k)\in t_{i}\in T_{11}} x_{i} \geq 1,  ~~ \forall~ k\mid s_k \in S_{11} 
\end{align*}
\vspace{-.5cm}
\begin{align*}
&~~~x_{i} \in \{0,1\} ~~~~\forall ~ i \mid t_{i}\in T_{11}  
\end{align*}
\end{minipage}
\vline
\begin{minipage}{.5\textwidth}
\begin{align*}
\mathtt{Z}_{12}: \text{ } \displaystyle{\min  \sum_{j \mid t_{j}\in T_{12}} y_{j}} \hspace{3.3cm}
\end{align*} 
\vspace{-.5cm}
\begin{align*}
~~~\text{s.t. } & \sum_{j\mid l(s_k)\in t_{j}\in T_{12}} y_{j} \geq 1,  ~~ \forall~ k\mid s_k \in S_{12}
\end{align*}
\vspace{-.5cm}
\begin{align*}
&~~~ y_{j} \in \{0,1\} ~~~~\forall ~  j \mid t_{j}\in T_{12}
\end{align*}

\end{minipage}
\vspace{.05cm}

Then by an analysis identical to Gaur et al. \cite{Gaur_2002}, we conclude that, $OPT_{11}^F + OPT_{12}^F \leq 2 OPT_1^F$. A similar analysis for $\mathtt{Z}_2$ leads to the following $OPT_{21}^F + OPT_{22}^F \leq 2 OPT_2^F$.

\noindent {\bf Step 3:} In step 2, there are four $ILP$'s, $\mathtt{Z}_{11}$ and $\mathtt{Z}_{12}$ corresponding to $\mathtt{Z}_{1}$ and $\mathtt{Z}_{21}$ and $\mathtt{Z}_{22}$ corresponding to $\mathtt{Z}_{2}$. Observe that, $\mathtt{Z}_{11}$ is the problem of covering the left end-points of the segments $S_{11}\subseteq S_{1}$ by those unit squares which intersect even indexed lines. Similarly, $\mathtt{Z}_{12}$ is the problem of covering the left end-points of the segments $S_{12}\subseteq S_{1}$ by those unit squares which intersect odd indexed lines.

Now, focus our attention on $\mathtt{Z}_{11}$. Since the squares intersected by an even indexed line $\ell_i$ is independent of the squares intersected by some other even indexed line $\ell_j$ and the points covered by the squares intersected by $\ell_i$ are not covered by the squares intersected by $\ell_j$ ($i \neq j$), we can split $\mathtt{Z}_{11}$ into different {\it ILP}s corresponding to each even indexed line, which can be independently solved. 

Now consider $\mathtt{Z}_{11}^\xi$ corresponding to the line $\ell_\xi$, where $\xi$ is even. 
Let $S_{11}^\xi$ be the set of segments whose left end-points are to be covered by the set of squares $T_{11}^\xi$, which are intersected by the line $\ell_\xi$. We split $S_{11}^\xi$ into disjoint sets $S_{111}^\xi$ and $S_{112}^\xi$, where $S_{111}^\xi$ (resp. $S_{112}^\xi$) are the set of all segments above (resp. below) the line $\ell_\xi$. The objective is to cover the left end-points of the members $S_{111}^\xi$ (resp. $S_{112}^\xi$) by minimum number of squares.  
Thus, the {\it ILP} of $\mathtt{Z}_{11}^\xi$ can be written as, 

\vspace{-5mm}
\begin{align*}
\mathtt{Z}_{11}^\xi: \text{ } \displaystyle{\min \sum_{i\mid t_{i}\in T_{11}^\xi} x_{i}} \hspace{3.3cm}
\end{align*} 
\vspace{-.5cm}
\begin{align*}
 ~~~\text{s.t. } &\sum_{i \mid l(s_k)\in t_{i}\in T_{11}^\xi} x_{i} \geq 1  ~~~\forall~ k\mid s_k \in S_{111}^\xi, ~~\text{ and }  
  \sum_{i \mid l(s_k)\in t_{i}\in T_{11}^\xi} x_{i} \geq 1  ~~~\forall~ k\mid s_k \in S_{112}^\xi 
\end{align*}
\vspace{-.5cm}
\begin{align*}
&~~~x_{i} \in \{0,1\} ~~~~\forall ~ i \mid t_{i}\in T_{11}^\xi  
\end{align*}

\vspace{-0.1in}
Again, since the sets $S_{111}^\xi$ and $S_{112}^\xi$ are disjoint, we may consider 
the following two $ILP$'s, $\mathtt{Z}_{111}^\xi$ and $\mathtt{Z}_{112}^\xi$ as follows.

\vspace{-0.1in}
\begin{minipage}{.5\textwidth}
\begin{align*}
\mathtt{Z}_{111}^\xi: \text{ } \displaystyle{\min \sum_{i\mid t_{i}\in T_{11}} x_{i}} \hspace{3.3cm}
\end{align*} 
\vspace{-.5cm}
\begin{align*}
~~~\text{s.t. } &\sum_{i \mid l(s_k)\in t_{i}\in T_{11}^\xi} x_{i} \geq 1  ~~ \forall~ k\mid s_k \in S_{111}^\xi
\end{align*}
\vspace{-.5cm}
\begin{align*}
&~~~x_{i} \in \{0,1\} ~~~~\forall ~ i \mid t_{i}\in T_{11}^\xi  
\end{align*}
\end{minipage}
\vline
\begin{minipage}{.5\textwidth}
\begin{align*}
\mathtt{Z}_{112}^\xi: \text{ } \displaystyle{\min \sum_{i\mid t_{i}\in T_{11}} x_{i}} \hspace{3.3cm}
\end{align*} 
\vspace{-.5cm}
\begin{align*}
~~~\text{s.t. } &\sum_{i \mid l(s_k)\in t_{i}\in T_{11}^\xi} x_{i} \geq 1  ~~ \forall~ k\mid s_k \in S_{112}^\xi 
\end{align*}
\vspace{-.5cm}
\begin{align*}
&~~~x_{i} \in \{0,1\} ~~~~\forall ~ i \mid t_{i}\in T_{11}^\xi  
\end{align*}
\end{minipage}

Let $\widetilde{x^*}$ be an optimal fractional solution of  $\overline{\mathtt{Z}}_{11}^\xi$. Clearly, $\widetilde{x^*}$ satisfies all the constraints in both $\overline{\mathtt{Z}}_{111}^\xi$ and $\overline{\mathtt{Z}}_{112}^\xi$. Also, it is observe that $OPT_{111}^{\xi F} \leq OPT_{11}^{\xi F}$ and $OPT_{112}^{\xi F} \leq OPT_{11}^{\xi F}$. Combining, we conclude that $OPT_{111}^{\xi F} + OPT_{112}^{\xi F} \leq 2 OPT_{11}^{\xi F}$.

A similar equation can be shown for $\mathtt{Z}_{12}^\xi$ as follows: $OPT_{121}^{\xi F} + OPT_{122}^{\xi F} \leq 2 OPT_{12}^{\xi F}$.

Finally, we have the following four equations, 
\begin{enumerate}
\item $\sum_{\xi \text{ even}}OPT_{111}^{\xi F} + \sum_{\xi \text{ even}} OPT_{112}^{\xi F} \leq 2 \sum_{\xi \text{ even}} OPT_{11}^{\xi F} \leq 2OPT_{11}^{F}$, 
\item $\sum_{\xi \text{ even}}OPT_{121}^{\xi F} + \sum_{ \xi \text{ even}} OPT_{122}^{\xi F} \leq 2 \sum_{ \xi \text{ even}} OPT_{12}^{\xi F} \leq 2OPT_{12}^{F}$, 
\item $\sum_{ \xi \text{ odd}}OPT_{211}^{\xi F} + \sum_{ \xi \text{ odd}} OPT_{212}^{\xi F} \leq 2 \sum_{ \xi \text{ odd}} OPT_{21}^{\xi F} \leq 2OPT_{21}^{F}$, and 
\item $\sum_{ \xi \text{ odd}}OPT_{221}^{\xi F} + \sum_{ \xi \text{ odd}} OPT_{222}^{\xi F} \leq 2 \sum_{ \xi \text{ odd}} OPT_{22}^{\xi F} \leq 2OPT_{22}^{F}$.
\end{enumerate}

\noindent {\bf Step 4:} In this step we apply Lemma \ref{bansal-lemma} independently on each of $\mathtt{Z}_{111}^\xi$, $\mathtt{Z}_{112}^\xi$, $\mathtt{Z}_{121}^\xi$, $\mathtt{Z}_{122}^\xi$ where $\xi$ is even, and $\mathtt{Z}_{211}^\xi$, $\mathtt{Z}_{212}^\xi$, $\mathtt{Z}_{221}^\xi$, $\mathtt{Z}_{222}^\xi$ where $\xi$ is odd to get the following eight equations.
\begin{tabbing}
(vii) \= xxxxxxxxxxxxxxxxxxxxxx\= (viii) \= xxxxxxxxxxxxxxxxxxxxxx\= (iii) \= xxxxxxxxxxxxxxxxxxxxxx \= \kill 
(i)\> $OPT_{111}^{\xi I} \leq 2 OPT_{111}^{\xi F}$, \> (ii) \> $OPT_{112}^{\xi I} \leq 2 OPT_{112}^{\xi F}$, \> (iii) \> $OPT_{121}^{\xi I} \leq 2 OPT_{121}^{\xi F}$, \\
(iv)\> $OPT_{122}^{\xi I} \leq 2 OPT_{122}^{\xi F}$,  \> (v) \> $OPT_{211}^{\xi I} \leq 2 OPT_{211}^{\xi F}$, \> (vi) \> $OPT_{212}^{\xi I} \leq 2 OPT_{212}^{\xi F}$,\\
(vii) \> $OPT_{221}^{\xi I} \leq 2 OPT_{221}^{\xi F}$, \> (viii) \> $OPT_{222}^{\xi I} \leq 2 OPT_{222}^{\xi F}$.
\end{tabbing}

\noindent {\bf Approximation factor:} Now combining all the inequalities of Step 4, we have
\begin{tabbing}
xx\= \kill
\>$\sum_{\xi \text{ even}} OPT_{111}^{\xi I} + \sum_{ \xi \text{ even}} OPT_{112}^{\xi I} + \sum_{ \xi \text{ even}} OPT_{121}^{\xi I} + \sum_{ \xi \text{ even}} OPT_{122}^{\xi I}$\\
$+$\> $\sum_{ \xi \text{ odd}} OPT_{211}^{\xi I} + \sum_{ \xi \text{ odd}} OPT_{212}^{\xi I} + \sum_{ \xi \text{ odd}} OPT_{221}^{\xi I} + \sum_{ \xi \text{ odd}} OPT_{222}^{\xi I}$\\
$\leq$\>$ 2 (\sum_{ \xi \text{ even}} OPT_{111}^{\xi F} + \sum_{ \xi \text{ even}} OPT_{112}^{\xi F} + \sum_{ \xi \text{ even}} OPT_{121}^{\xi F} + \sum_{ \xi \text{ even}} OPT_{122}^{\xi F}$ \\
$+$\> $\sum_{ \xi \text{ odd}} OPT_{211}^{\xi F} + \sum_{ \xi \text{ odd}} OPT_{212}^{\xi F} + \sum_{ \xi \text{ odd}} OPT_{221}^{\xi F} + \sum_{ \xi \text{ odd}} OPT_{222}^{\xi F}) $
\end{tabbing}

$\leq$  $ 16~ OPT_0^I$,  by applying the inequalities of Step 3, Step 2 and Step 1 in this order. 

\begin{theorem}  There exists a factor $16$ approximation algorithm for \textit{DCSUS-ARB} problem that runs in polynomial time.
\end{theorem}

\newpage
\appendix
\section{Appendix}
\subsection{Conflict-free covering of points on a real line using unit intervals} \label{ABC}

Arkin et al \cite{Arkin2015} proved that the discrete version of conflict-free covering problem with intervals where the points are on real line is \np-complete, using  
 a reduction from the vertex cover problem. We slightly  modify this reduction as follows. Let $G(V,E)$ be a given graph with each vertex having distinct $x$-coordinate. Project each vertex $v \in V$ at a point $x_v$ on the $x$-axis such that the distance between each pair of consecutive projections of vertices is strictly greater than 1. Now for each vertex $v \in V$, we take an unit interval $I_v$  such that the mid point of $I_v$ coincides with $x_v$ on $x$-axis. Now, corresponding to each edge $e=(u,v)$, we take a single horizontal line segment connecting $x_u$ and $x_v$. 
 Thus we have an instance of the \textit{DCSUS-ARB} problem where each interval represents a unit square with bottom boundary aligned to that interval. It is easy to show that any feasible solution of the \textit{DCSUS-ARB} problem is the solution of the given vertex cover problem, and vice verse. Thus, we have the following observations.

 \begin{observation}~
 \begin{itemize}
  \item[$\bullet$] \textit{DCSUS-ARB} problem is \np-complete even if the given segments are on a real line.
  \item[$\bullet$] Since the distance between any two consecutive projections $(x_u,x_v)$ is greater than 1, an unit interval cannot cover end-points of segments on both $x_u$ and $x_v$.  Thus in order to cover the segments corresponding to the edges in $G$, we need to choose the unit squares (i.e., unit intervals) from $I_v$s' as stated above. Thus, \textit{CCSUS-ARB} is also \np-complete.
  \item[$\bullet$] Since vertex cover cannot be approximated better than factor $2$ \cite{Khot2008}, both \textit{CCSUS-ARB} and \textit{DCSUS-ARB} cannot be approximated better than factor $2$ even if the segments are on a real line.
 \end{itemize}
\end{observation}

\subsection{Analysis of the count on the available positions for placing unit squares in a cell of size $k \times k$} \label{appendix1}
For a pair of end-points $\alpha$ and $\beta$ of 
two different segments $s_i$ and $s_j$ respectively, $j \neq i$, with $d_\infty(\alpha,\beta)$ ($L_\infty$ 
distance of $\alpha,\beta$) $\leq 1$), we may have two distinct positions of placing an unit square. These squares 
will be referred to as {\it type-1} positions. However, if for an end-point $\alpha$ of some segment $s_i$, there is no end-point $\beta$ of some other segment such 
that $d_\infty(\alpha,\beta) \leq 1$, then an unit square covering only $s_i$ is required, and it can be placed anywhere containing an end-point of $s_i$. For the 
sake of simplicity we will take a unit-square whose top-left corner is anchored at such a point, and will refer 
such a position as {\it type-2} position. Thus, for each pair of segments $s_i,s_j$ inside $C$, we may get at most two {\it type-1} positions, and 
for each segment $s_i$, we consider two {\it type-2} positions corresponding to the two end-points of $s_i$. Thus, the count stated in the proof of Lemma \ref{ptas-lemma} follows.
\newpage

\subsection{Proof of Lemma \ref{bansal-lemma}} \label{appendix2}
Consider the {\it ILP} formulation of \textit{Restricted-Point-Cover} problem 

\noindent \begin{minipage}{.5\textwidth}
\begin{align*}
\mathtt{Z_{RPC}}: \text{ } \displaystyle{\min \sum_{r \in {\cal R}}  x_r} \hspace{3.3cm}
\end{align*} 
\vspace{-.5cm}
\begin{align*}
 ~~~\text{s.t.} &\sum_{r \mid p\in r ~\&~ r \in {\cal R}} x_r \geq 1,  ~~\forall~ p\in P; ~~~~ x_r\in \{0,1\} ~~\forall ~ r \in {\cal R}
\end{align*}
\vspace{-.3cm}
\end{minipage}

Let us consider the following algorithm. It uses two sets $LB$ and $OUTPUT$ as in Algorithm \ref{algo-fact-3}. 

Procedure {\bf RPC}($P$, $\cal R$)
\begin{itemize}
\item[1.] Consider the point in $P$ having maximum $y$-coordinate.
\item[2.] Identify the subset ${\cal R}_p$ of rectangles in $\cal R$ that cover $p$. 
\item[3.] Choose the $r_\ell \in {\cal R}_p$ having left-most left boundary, and 
$r_r \in {\cal R}_p$ having right-most right boundary.
\item[4.] Include $p \in LB$, and $r_\ell,r_r$ in $OUTPUT$.\\
 Delete all the rectangles ${\cal R}_p$ from $\cal R$; \\
 Delete the subset of points $P_p \subseteq P$ that are covered by 
 the members in ${\cal R}_p$.\\
 This splits (i) the remaining points into disjoint subsets $P_1$ and $P_2$ to the left of  the left boundary of $r_\ell$ and to the right of the right boundary of $r_r$, and (ii) the remaining rectangles into disjoint subsets ${\cal R}^1$ and ${|cal R}^2$ 
 that are to the left of  the left boundary of $r_\ell$ and to the right of the right boundary of $r_r$.
\item[5.] If $P_1 \neq \emptyset$ then call {\bf RPC}($P_1$, ${\cal R}^1$), and \\
   If $P_2 \neq \emptyset$ then call {\bf RPC}($P_2$, ${\cal R}^2$).
\item[6.] Report $OUTPUT$
\end{itemize} 

Observe that, in the optimal solution of the {\it LP} $\overline{\mathtt{Z}}_{RPC}$ corresponding to the {\it ILP} $\mathtt{Z_{RPC}}$, if $p \in LB$ then 
the constraint $\displaystyle{\sum_{r \in {\cal R}_p} x_r \geq 1}$ is satisfied. In other words, the sum of the variables corresponding to the rectangles in ${\cal R}_p$ is at least 1. Also, these variables 
do not occur in the constraint for any other point $p' \in LB$. In order to 
satisfy the constraints corresponding to the points $p \in P\setminus LB$, some 
other variables may be assigned positive values. Thus, $OPT_{RPC}^F \geq LB$.
In our solution of the {\it ILP} $\mathtt{Z_{RPC}}$, we have chosen variables for setting value 1 such that the constraint corresponding to each $p \in LB$ attains the value 2. All other variables are set to 0. Thus, $OPT_{RPC}^I =2LB$. Thus, the lemma is proved.

\end{document}